\documentclass[letterpaper,11pt]{article}

\usepackage[margin=1in]{geometry}
\usepackage{amsmath}
\usepackage{amssymb}
\usepackage{amsfonts}
\usepackage{amsthm}
\usepackage{array}
\usepackage{float}
\usepackage{bm}

\newcommand{\bcomment}[1]{}
\newcommand{\TODO}[1]{}

\newcommand{\pth}[1]{\left(#1\right)}

\newcommand{\R}{\mathbb{R}}
\newcommand{\N}{\mathbb{N}}

\newcommand{\E}{\operatorname{\mathbb{E}}}

\newcommand{\twobytwo}[4]{\begin{pmatrix} #1 & #2 \\ #3 & #4 \end{pmatrix}}

\newcommand{\bfp}{\mathbf{p}}

\newcommand{\bfG}{\mathbf{G}}
\newcommand{\bfM}{\mathbf{M}}
\newcommand{\bfS}{\mathbf{S}}
\newcommand{\bfX}{\mathbf{X}}
\newcommand{\bfY}{\mathbf{Y}}
\newcommand{\bfmu}{\bm{\mu}}

\newcommand{\mcM}{\mathcal{M}}

\newcommand{\mcU}{\mathcal{U}}
\newcommand{\mcV}{\mathcal{V}}

\newcommand{\mcX}{\mathcal{X}}
\newcommand{\mcY}{\mathcal{Y}}

\newcommand{\eq}{\mathrel{\phantom{=}}}
\newcommand{\tr}{\operatorname{tr}}

\newcommand{\multisetds}[2]{\bigg(\kern-.4em\binom{#1}{#2}\kern-.4em\bigg)}
\newcommand{\multisetin}[2]{\big(\kern-.3em\binom{#1}{#2}\kern-.3em\big)}
\newcommand{\multisetix}[2]{\left(\kern-.2em\binom{#1}{#2}\kern-.2em\right)}

\newcommand{\eql}[1]{\overset{(#1)}{=}}
\newcommand{\leql}[1]{\overset{(#1)}{\leq}}

\DeclareMathOperator*{\argmin}{argmin}
\DeclareMathOperator*{\argmax}{argmax}

\usepackage{color}
\usepackage{tikz}
\usepackage{bbm}
\usepackage{authblk}

\newtheorem{theorem}{Theorem}

\newtheorem{lemma}{Lemma}

\newtheorem{definition}{Definition}

\newcommand{\Aut}{\operatorname{Aut}}
\newcommand{\id}{\operatorname{id}}

\newcommand{\ER}{\operatorname{ER}}
\newcommand{\ErRe}{Erd\H{o}s-R\'{e}nyi }

\newcommand{\pmar}{p_{1*}p_{*1}}
\newcommand{\G}[1]{G_{#1}}%{G_a \cap_{#1} G_b}
\newcommand{\gfM}{A_{\mu,\mu^*}}
\newcommand{\ttl}{\tilde{t}}%{\widetilde{t}}
\newcommand{\onevec}{\vec{1}}
\newcommand{\indset}{\mathbbm{1}}

\begin{document}
\title{Partial Recovery of Erd\H{o}s-R\'{e}nyi Graph Alignment via $k$-Core Alignment}
%\author{}
\author[1]{Daniel Cullina\thanks{dcullina@princeton.edu}}
\author[2]{Negar Kiyavash\thanks{negar.kiyavash@ece.gatech.edu}}
\author[1]{Prateek Mittal\thanks{pmittal@princeton.edu}}
\author[1]{H. Vincent Poor\thanks{poor@princeton.edu}}
\affil[1]{Princeton University, Department of Electrical Engineering}
\affil[2]{Georgia Tech, Department of Electrical and Computer Engineering and Department of Industrial \& Systems Engineering}
\date{}
\maketitle
\begin{abstract}
  We determine information theoretic conditions under which it is possible to \emph{partially} recover the alignment used to generate a pair of sparse, correlated \ErRe graphs.
  To prove our achievability result, we introduce the $k$-core alignment estimator.
  This estimator searches for an alignment in which the intersection of the correlated graphs using this alignment has a minimum degree of $k$.
  We prove a matching converse bound.
  As the number of vertices grows, recovery of the alignment for a fraction of the vertices tending to one is possible when the average degree of the intersection of the graph pair tends to infinity.
  It was previously known that exact alignment is possible when this average degree grows faster than the logarithm of the number of vertices.
\end{abstract}
%\newpage

Graph alignment, or graph matching, is the problem of finding a correspondence between the vertex sets of a pair of graphs using structural information from the graphs.
It can be thought of as the noisy generalization of the graph isomorphism problem.
Graph matching has applications in the privacy of social network data, the analysis of biological protein interaction networks, and in computer vision.

We consider the graph matching problem for random graphs that have been generated in a correlated way, so there is a planted ground-truth alignment of their vertices.
In this setting, the combinatorial optimization problem of maximizing edge overlap is also the maximum a posteriori estimator.

\subsection{Related work}
A number of authors have worked to determine the information theoretic or statistical conditions under which graph alignments can be recovered by any algorithm.
Wright determined the conditions under which an \ErRe graph has a trivial automorphism group, or equivalently under which the isomorphism recovery problem has a unique solution \cite{wright_graphs_1971}.
Pedarsani and Grossglauser obtained achievability conditions for exact recovery in the noisy case \cite{pedarsani_privacy_2011}.
Cullina and Kiyavash obtained matching achievability and converse conditions for exact recovery \cite{cullina_exact_2017,cullina_improved_2016}.
Kazemi, Yartseva, and Grossglauser considered alignment of graphs with overlapping but not identical vertex sets \cite{kazemi_when_2015}.
Shirani, Garg, and Erkip found an achievability condition for partial recovery with a small number of errors was obtained \cite{shirani_typicality_2018}.
In all of these cases, the explicit or implicit algorithms require exponential time in the number of vertices.
Cullina, Mittal, and Kiyavash obtained analogous limits for the alignment recovery problem for correlated databases \cite{cullina_fundamental_2018}.
In this case, maximum a posteriori estimation can be done efficiently.

A number of practically motivated and efficient algorithms have been proposed
\cite{ji_your_2015,lee_blind_2017,backstrom_wherefore_2007,korula_efficient_2014,malod-dognin_l-graal:_2015,kuchaiev_topological_2010,singh_global_2008,tian_tale:_2008,zhang_sapper:_2010}.
These have largely been empirically evaluated on a mix of real and synthetic datasets.
It is common for these algorithms to return partial matchings of the vertex sets for some graph pairs.

A few efficient algorithms that require some form of initial side information have been rigorously analyzed.
Yartseva and Grossglauser used a percolation algorithm to obtain a graph alignment starting with some matched pairs of seed vertices \cite{yartseva_performance_2013}.
A number of other works have investigated seeded matching \cite{kazemi_growing_2015,pedarsani_bayesian_2013,lyzinski_seeded_2014}.
Feizi et al. used a spectral method to recover an alignment of dense graphs with side information restricting the set of possible alignments \cite{feizi_spectral_2016}.
Lyzinski et al. explored the limitations of some convex programming methods, which have presented a barrier to the development of efficient algorithms \cite{lyzinski_graph_2016}.

Very recently, provably correct quasi-polynomial time algorithms have been obtained.
Barak, Chou, Lei, Schramm, and Sheng search for appearance of particular polylogarithmic-sized subgraphs in both graphs \cite{barak_nearly_2018}.
Mossel and Xu use seeds more efficiently than previous algorithms, creating a signature vertex based on the set of seeds in a large neighborhood of the vertex.
The number of seed pairs required is small enough that they can be guessed, yielding an algorithm that does not require side information \cite{mossel_seeded_2018}.

We intentionally use the terminology ``planted alignment'' in analogy with ``planted clique'', ``planted dense subgraph'', ``planted coloring'', and ``planted partition''.
For these settings, there are several basic problems.
One is to find the statistical or information theoretic limits of exact recovery, i.e. the conditions under which an algorithm with unlimited resources can with high probability recover the hidden structure with no errors.
Another is to find the information theoretic limits of detection, i.e. the conditions under which an object with a planted structure can be distinguished from an object without one.
Finally, there are the conditions under which efficient algorithm can succeed at these tasks.
There is a large body of work using spectral algorithms, message passing algorithms, and semidefinite optimization to efficiently recover planted structures.
See the surveys of Moore \cite{moore_computer_2017}, Abbe \cite{abbe_community_2017}, and Wu and Xu \cite{wu_statistical_2018} for an overview.

In the case of recovering a planted alignment, finding the information-theoretic limits of exact recovery, often the easiest of the standard problems to resolve, is already challenging.
In this paper, we investigate the information-theoretic limits of a problem in between exact recovery and detection: recovery of almost all of a planted alignment with one-sided error.

\section{Model}
\subsection{Notation}
A binary relation $\mu \subseteq \mcU_a \times \mcU_b$ is a matching if each $i \in \mcU_a$ and $j \in \mcU_b$ appears in at most one pair in $\mu$.
%For a matching $\mu \subseteq \mcU_a \times \mcU_b$,
Define the functions $\alpha : 2^{\mcU_a \times \mcU_b} \to 2^{\mcU_a}$ and $\beta : 2^{\mcU_a \times \mcU_b} \to 2^{\mcU_b}$ that find the left and right support of a binary relation: 
\begin{align*}  
  \alpha(\mu) &= \{i \in \mcU_a : \exists j \in \mcU_b \, . \,  (i,j) \in \mu\}\\
  \beta(\mu) &= \{j \in \mcU_b : \exists i \in \mcU_a \, . \,  (i,j) \in \mu\}.
\end{align*}
%That is, $\alpha$ and $\beta$ give the sets of elements that appear somewhere in $\mu$. 
A matching $\mu \subseteq \mcU_a \times \mcU_b$ is a bijection if $\alpha(\mu) = \mcU_a$ and $\beta(\mu) = \mcU_b$.

Let $\wedge$ be the minimum or ``and'' binary operator on $\{0,1\}$.
Let $[n]$ denote the set $\{0,\cdots,n-1\}$.
For a set $\mcU$, let $\binom{\mcU}{2}$ be the set of unordered pairs of elements of $\mcU$. 
Represent a labeled graph on a vertex set $\mcU$ by its edge indicator function $G : \binom{\mcU}{2} \to [2]$.
For a graph $G$, let $V(G)$ and $E(G)$ be the vertex and edge sets respectively. 
%The edge set of a graph $G$ is $E(G) \subseteq \binom{V(G)}{2}$.
%We use standard notation to express bounds on the asymptotic behavior a function: $o(f(n))$, $\mathcal{O}(f(n))$, $\Omega(f(n))$, and $\omega(f(n))$.
Throughout, we use boldface letters for random objects and lightface letters for fixed objects.

\subsection{Correlated graphs}
The correlated \ErRe graph model has been used in much of the work on alignment recovery for random graphs, beginning with Pedarsani and Grossglauser \cite{pedarsani_privacy_2011}.
The idea is simple: we have two graphs $G_a$ and $G_b$ whose marginal distributions are \ErRe.
Under the true vertex matching, each edge random variable in $G_a$ is aligned with some edge random variable in $G_b$.
These aligned pairs of edge random variables have some joint distribution and this is the only source of correlation between the graphs.

To formalize this, we need the following definition.
\begin{definition}[Lifted matching]
A matching $\mu \subseteq \mcU_a \times \mcU_b$ gives rise to a lifted matching $\ell(\mu) \subseteq \binom{\mcU_a}{2} \times \binom{\mcU_b}{2}$,
\[
  \ell(\mu) = \left\{(\alpha(w),\beta(w)) : w \in \binom{\mu}{2}\right\} = \left\{(\{u_a,v_a\},\{u_b,v_b\}) : \{(u_a,u_b),(v_a,v_b)\} \in \binom{\mu}{2}\right\}.
\]
\end{definition}

\begin{definition}
  The distribution of random variables $(\bfX_a,\bfX_b) \in \{0,1\}^2$ is fully specified by a matrix of parameters $p \in \R^{\{0,1\} \times \{0,1\}}$, where
  $P_{\bfX_a,\bfX_b}(i,j) = p_{ij}$.
  In this case, we say that $X_a$ and $X_b$ have a correlated Bernoulli distribution with parameter $p$.

  For a matching $\bfmu \subseteq \mcU_a \times \mcU_b$, we define the correlated \ErRe distribution on pairs of graphs $\bfG_a : \binom{\mcU_a}{2} \to \{0,1\}$ and $\bfG_b: \binom{\mcU_b}{2} \to \{0,1\}$, denoted $\ER(\bfmu,p)$, as follows.
  For each $(w_a,w_b) \in \ell(\bfmu)$, $(\bfG_a(w_a), \bfG_b(w_b))$ have a correlated Bernoulli distribution with parameter $p$ and these random variables are mutually independent.
  That is,
  \[
    P_{\bfG_a,\bfG_b|\bfmu}(G_a,G_b|\mu) = \prod_{(w_a,w_b) \in \ell(\mu)} P_{\bfX_a,\bfX_b}(G_a(w_a),G_b(w_b)).
  \]
\end{definition}

Because $l(\mu)$ is a matching, each $w_a \in \binom{\mcU_a}{2}$ appears in exactly one pair $(w_a,w_b) \in l(\mu)$.
%Equivalently, we can interpret l(\mu) as a bijective function and find w_b = l(\mu)(w_a).
For a pair $(w_a,w_b) \not\in \ell(\mu)$, $G_a(w_a)$ is independent of $G_b(w_b)$.

If $p_{11}p_{00} > p_{10}p_{01}$, then these distributions have \emph{positive correlation}.
We will only consider positively correlated graphs in this paper.

\subsection{Estimating a planted alignment}
\label{section:estimation}
We consider the following estimation problem.
Let $|\mcU_a| = |\mcU_b| = n$ and let $\bfmu$ be a uniformly random bijection between $\mcU_a$ and $\mcU_b$.
Let $(\bfG_a,\bfG_b) \sim \ER(\bfmu,p)$.

The most stringent recovery requirement, $\hat{\bfmu} = \bfmu$ or exact recovery, was addressed by Cullina and Kiyavash \cite{cullina_exact_2017}.
Their precise results are discussed in Section~\ref{section:results}.
In that case, there is a clear definition of the optimal estimator: the maximum a posteriori (MAP) estimator: $\hat{\bfmu}(G_a,G_b) = \argmax_{\hat{\mu}} \Pr[\bfmu = \hat{\mu}|\bfG_a = G_a,\bfG_b = G_b]$.
Because $\bfmu$ has a uniform distribution, by Bayes theorem $\hat{\bfmu}(G_a,G_b) = \argmax_{\hat{\mu}} \Pr[\bfG_a = G_a,\bfG_b = G_b|\bfmu = \hat{\mu}]$.

This estimator is closely related to the \emph{aligned intersection} of a pair of graphs.
Let $G_a$ and $G_b$ be graphs and let $\mu$ be a matching between their vertex sets.
Then $\mu$ provides an alignment of the subgraphs $G_a[\alpha(\mu)]$ and $G_b[\beta(\mu)]$.
Using this alignment, we can compute the intersection of these two subgraphs.
The natural vertex set for this intersection graph is $\mu$.
We formalize this construction as follows.
\begin{definition}
  Let $G_a$ and $G_b$ be graphs and let $\mu \subseteq V(G_a) \times V(G_b)$ be a matching.
  Define $G_a \wedge_\mu G_b$, the aligned intersection of $G_a$ and $G_b$, to be the graph on the vertex set $\mu$ such that
  \begin{align*}
%    E(G_a \cap_\mu G_b) &= \left\{ \{(u_a,u_b),(v_a,v_b)\} \in \binom{\mu}{2} : \{u_a,v_a\} \in E(G_a) \wedge \{u_b,v_b\} \in E(G_b) \right\}\\
    (G_a \wedge_\mu G_b)&: \binom{\mu}{2} \to \{0,1\}\\
    (G_a \wedge_\mu G_b)(\{(u_a,u_b),(v_a,v_b)\}) &= G_a(\{u_a,v_a\}) \wedge G_b(\{u_b,v_b\})
  \end{align*}
  or equivalently
  \[
    (G_a \wedge_\mu G_b)(w) = G_a(\alpha(w)) \wedge G_b(\beta(w)).
  \]  
\end{definition}
Cullina and Kiyavash~\cite{cullina_improved_2016} observed that for a bijection $\mu$,
\[
  \Pr[\bfG_a = G_a,\bfG_b = G_b|\bfmu = \mu] \propto \left(\frac{p_{11}p_{00}}{p_{10}p_{01}}\right)^{|E(G_a \wedge_\mu G_b)|}
\]
where the constant of proportionality depends on $G_a$ and $G_b$ but not on $\mu$.
Thus, in the case of positive correlation, the MAP estimator is $\hat{\bfmu}(G_a,G_b) = \argmax_{\hat{\mu}} |E(G_a \wedge_{\hat{\mu}} G_b)|$.

Herein we consider partial recovery of $\bfmu$ using $(\bfG_a,\bfG_b)$, which is interesting when exact recovery is impossible.
In particular, we would like to match some of the vertices of $\bfG_a$ to the corresponding vertices in $\bfG_b$ without any errors.
This means that we want an estimator $\hat{\bfmu}$ such that $\hat{\bfmu} \subseteq \bfmu$ and $|\hat{\bfmu}|$ is as large as possible.
We are interested in estimators that satisfy these conditions with probability $1-o(1)$.

For a partial matching $\mu'$,
\[
  \Pr[\bfmu \supseteq \mu'|\bfG_a = G_a,\bfG_b = G_b] = \sum_{\mu \supseteq \mu', |\mu| = n} \Pr[\bfmu = \mu'|\bfG_a = G_a,\bfG_b = G_b].
\]
There are two natural generalization of the MAP estimator for partial recovery.
The first fixes $n'$, the size of the estimated matching, and selects $\hat{\mu}$ that maximizes $\Pr[\bfmu \supseteq \hat{\mu}|\bfG_a = G_a,\bfG_b = G_b]$.
The second fixes $\epsilon$, an error budget, and selects a $\hat{\mu}$ satisfying $\Pr[\bfmu \supseteq \hat{\mu}|\bfG_a = G_a,\bfG_b = G_b] \geq 1-\epsilon$ and maximizing $|\hat{\mu}|$.
Neither of these are particularly straightforward to analyze, so we introduce the $k$-core alignment estimator.

%We will primarily have two graph random variables, $\bfG_a$ and $\bfG_b$, and for any matching $\mu$ we will abbreviated their aligned intersection $\bfG_a \wedge_{\mu} \bfG_b$ as $\bfG_{\mu}$.
\subsection{$k$-cores and $k$-core alignments}

\label{subsection:core}
Let $\delta(G)$ be the minimum degree of $G$ and for $S \subseteq V(G)$, let $G[S]$ be the subgraph of $G$ induced by $S$.
We adopt the convention that for the null graph, i.e. $G$ such that $V(G) = \varnothing$, $\delta(G) = \infty$.
Thus there is always some $S \subseteq V(G)$ such that $\delta(G[S]) \geq k$.
If $\delta(G[S]) \geq k$ and $\delta(G[S']) \geq k$, then $\delta(G[S \cup S']) \geq k$.
Thus there is a unique maximum among the sets that induce subgraphs with minimum degree at least $k$.
The subgraph induced by this set is the $k$-core of $G$ \cite{bollobas_evolution_1984}.
\footnote{When every nonempty induced subgraph of a graph $G$ has a minimum degree less than $k$, some authors say that the $k$-core does not exist. In this case the $k$-core of $G$ is the null graph under our convention.}

We introduce the following related notion.
\begin{definition}
  A $k$-core alignment of $G_a$ and $G_b$ is a matching $\mu \subseteq V(G_a) \times V(G_b)$ such that $\delta(G_a \wedge_{\mu} G_b) \geq k$ and for all matchings $\mu' \supset \mu$, $\delta(G_a \wedge_{\mu'} G_b) < k$.
%  Let $\mathcal{M}_k(G_a,G_b)$ be the set of $k$-core alignments.
\end{definition}

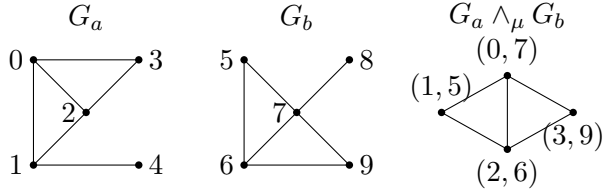
\begin{figure}
  \centering
  \begin{tikzpicture}[scale=0.7]
    \coordinate[label=left:{$0$}] (v0) at (0,5);
    \coordinate[label=left:{$1$}] (v1) at (0,3);
    \coordinate[label=left:{$2$}] (v2) at (1,4);
    \coordinate[label=right:{$3$}] (v3) at (2,5);
    \coordinate[label=right:{$4$}] (v4) at (2,3);

    \coordinate[label=left:{$5$}] (v5) at (4,5);
    \coordinate[label=left:{$6$}] (v6) at (4,3);
    \coordinate[label=left:{$7$}] (v7) at (5,4);
    \coordinate[label=right:{$8$}] (v8) at (6,5);
    \coordinate[label=right:{$9$}] (v9) at (6,3);

    \coordinate[label=above:{$(0,7)$}] (v07) at (9,4.7);
    \coordinate[label=above:{$(1,5)$}]  (v15) at (7.75,4);
    \coordinate[label=below:{$(2,6)$}] (v26) at (9,3.3);
    \coordinate[label=below:{$(3,9)$}] (v39) at (10.25,4);

%    \coordinate[label=above:{$(0,7)$}] (v07) at (3,1.7);
%    \coordinate[label=left:{$(1,5)$}] (v15) at (1.75,1);
%    \coordinate[label=below:{$(2,6)$}] (v26) at (3,0.3);
%    \coordinate[label=right:{$(3,9)$}] (v39) at (4.25,1);

    \foreach \point in {v0,v1,v2,v3,v4,v5,v6,v7,v8,v9,v07,v15,v26,v39}
    \fill[black] (\point) circle (2pt);

    \draw (v0) -- (v1) -- (v2) -- (v3) -- cycle;
    \draw (v0) -- (v2);
    \draw (v1) -- (v4);

    \draw (v5) -- (v6) -- (v9) -- (v7) -- cycle;
    \draw (v6) -- (v7) -- (v8);

    \draw (v07) -- (v15) -- (v26) -- (v39) -- cycle;
    \draw (v07) -- (v26);

    \draw (1,5.8) node {$G_a$};
    \draw (5,5.8) node {$G_b$};
    \draw (9,5.8) node {$G_a \wedge_{\mu} G_b$};
%    \draw (1,2.5) node {$V(G_a) = \{0,1,2,3,4\}$};
%    \draw (5,2.5) node {$V(G_b) = \{5,6,7,8,9\}$};
%    \draw (3,-1.5) node {};

\end{tikzpicture}
\caption{The matching $\mu = \{(0,7),(1,5),(2,6),(3,9)\}$ is a 2-core alignment of $G_a$ and $G_b$: $\delta(G_a \wedge_{\mu} G_b) = 2$ and extending $\mu$ with $(4,8)$ leads to a minimum degree of 0.}
\label{fig:alignment}
\end{figure}

Figure~\ref{fig:alignment} illustrates the concepts of aligned intersection and $k$-core alignment.

\begin{definition}
  The $k$-core alignment estimator $\hat{\bfmu}_k$ selects a $k$-core alignment of $(\bfG_a,\bfG_b)$.
  If there is more than one $k$-core alignment, it makes an arbitrary choice.
  %, i.e. $\hat{\mu} = \argmax_{\mu \in \mathcal{M}_k(G_a,G_b)} |\mu|$.
\end{definition}

\subsection{Results}
\label{section:results}
We have the following results about the performance of the $k$-core alignment estimator.
\begin{theorem}[Achievability]
  \label{thm:main}
  Let $p$ satisfy the conditions
  \begin{align}
    p_{11} &\geq \omega\pth{\frac{1}{n}}\label{density}\\
    p_{11} &\leq \frac{1}{8e^3}\label{sparsity}\\
    \frac{p_{01}p_{10}}{p_{00}p_{11}} + p_{01} + p_{10} &\leq n^{-\Omega(1)}.\label{correlation} 
  \end{align}
  Then there is a choice of $k$ such that with probability $1-o(1)$, the $k$-core alignment estimator $\hat{\bfmu}_k$ satisfies $\hat{\bfmu}_k \subseteq \bfmu$ and $|\hat{\bfmu}_k| \geq n(1-o(1))$.
  That is, the estimator includes no incorrect pairs %($|\hat{\bfmu}_k \setminus \bfmu| = 0$)
  and almost all correct pairs. %($|\bfmu \setminus \hat{\bfmu}_k| = o(n)$).
\end{theorem}
Section~\ref{section:achievability} contains the proof.
The main condition of Theorem~\ref{thm:main}, \eqref{density}, requires $\bfG_a \wedge_{\bfmu} \bfG_b$ to have an average degree that grows with $n$.
Condition \eqref{sparsity} is very mild sparsity constraint on $\bfG_a \wedge_{\bfmu} \bfG_b$.
Condition \eqref{correlation} requires $\bfG_a$ and $\bfG_b$ to have sufficient positive correlation and to be mildly sparse. 
\begin{theorem}[Converse]
  \label{thm:converse}
  Let
  \begin{align}
    p_{11} & \leq \mathcal{O}\pth{\frac{1}{n}} \label{converse-density}\\
    \frac{p_{01}p_{10}}{p_{11}p_{00}} &< 1. \label{positive-correlation}
  \end{align}
  Then for any estimator $\hat{\bfmu}$ of $\bfmu$ given $(\bfG_a,\bfG_b)$ and any integer sequence $\epsilon(n) \leq o(n)$, the probability that $\hat{\bfmu} \subseteq \bfmu$ and $|\hat{\bfmu}| \geq n-\epsilon(n)$ is $o(1)$.
\end{theorem}
Section~\ref{section:converse} contains the proof.
Observe that the condition \eqref{converse-density} matches \eqref{density}.
Condition~\eqref{positive-correlation} is weaker than than \eqref{correlation}: the converse is valid for any amount of positive correlation.

For the exact recovery problem, Cullina and Kiyavash \cite{cullina_exact_2017} showed that if $p$ satisfies the conditions
\begin{align*}
  p_{11} &\geq \frac{\log n + \omega(1)}{n}\\
  p_{11} + p_{01} +p_{10} &\leq \mathcal{O}{\left(\frac{1}{\log n}\right)}\\
  \frac{p_{01}p_{10}}{p_{11}p_{00}} &\leq \mathcal{O}{\left(\frac{1}{(\log n)^3}\right)},
\end{align*}
then the maximum a posteriori estimator for $\bfmu$ given $(\bfG_a,\bfG_b)$ is correct with probability $1 - o(1)$.
Additionally, if $\bfp$ satisfies
\begin{equation*}
  p_{11} \leq \frac{\log n - \omega(1)}{n} \quad \text{and} \quad
  \frac{p_{01}p_{10}}{p_{11}p_{00}} < 1,
\end{equation*}
then any estimator for for $\bfmu$ given $(\bfG_a,\bfG_b)$ is correct with probability $o(1)$.
In other words, exact recovery of $\bfmu$ requires logarithmic average degree in the intersection graph while recovery of almost all of $\bfmu$ requires only a growing average degree.

\subsection{Product graphs}
The aligned intersection of $G_a$ and $G_b$ has another interpretation.
Let $G_a \times G_b$ be the tensor product of $G_a$ and $G_b$.
This is the graph with $V(G_a \times G_b) = V(G_a) \times V(G_b)$ and
\[
  (G_a \times G_b)(\{(u_a,u_b),(v_a,v_b)\}) = G_a(\{u_a,v_a\}) \wedge G_b(\{u_b,v_b\}).
\]
In other words, the adjacency matrix of $G_a \times G_b$ is the tensor product of the adjacency matrices of $G_a$ and $G_b$.
Then $G_a \wedge_{\mu} G_b = (G_a \times G_b)[\mu]$.

From this point of view, exact recovery of $\bfmu$ corresponds to finding a dense $n$-vertex subgraph inside the $n^2$-vertex graph $G_a \times G_b$.
This looks superficially like recovering a single dense community is a stochastic block model, a problem which has been extensively studied.
There are two important differences.
First, we only need to search over subgraphs induced by matchings, not all $n$-vertex subgraphs.
This does not significantly reduce the total number of candidate subgraphs, but it has a bigger effect on the number of subgraphs that are nearly equal to the true matching.
Second, the edge random variables in $G_a \times G_b$ are not jointly independent.
Because of this, bounding the probability of each error event requires some care.

The fact the we are searching for a subgraph of size $\sqrt{|V(G_a \times G_b)|}$ has potential implications for the computational tractability of this problem due to the planted clique hypothesis \cite{wu_statistical_2018} and associated conditional hardness results for planted dense subgraph problems.
However, in this paper, we focus only on information-theoretic thresholds.

Note that the $k$-core alignment of $G_a$ and $G_b$ is not the $k$-core of $G_a \times G_b$, which in general will be much larger and not induced by a matching.

\subsection{Proof outline}
Our achievability proof has the following structure.
First, we need to show that the true alignment of $\bfG_a$ and $\bfG_b$ yields a large $k$-core alignment.
This follows easily from known results about the $k$-core of an \ErRe graph.
The majority of our work is devoted to the second task: showing that this is the only $k$-core alignment.
For each matching $\mu \not\subseteq \mu^*$, we need to bound the probability that it is a $k$-core alignment of $G_a$ and $G_b$.
A large number of matchings can be immediately ruled by the maximality part of the definition of $k$-core alignment.
We define an error event for each of the remaining matching and use a union bound over them.
There are exponentially many of these error events, so it is crucial that our bound on error probability depends on the distance between the imposter matching and the true matching. 
This part of the argument is in Section~\ref{section:weak} and is summarized in Lemma~\ref{lemma:union-bound}.

In order for $\mu$ to be a $k$-core alignment, each vertex in $\bfG_a \wedge_{\mu} \bfG_b$ must have degree at least $k$.
This is much more difficult for vertices in $\mu$ that do not appear in the true matching $\mu^*$.
To bound the probability that $\mu$ is a $k$-core alignment, we consider the sum of the degrees of the vertices in $\mu \setminus \mu^*$.

The main technical task is to obtain large deviations upper bounds for these sum-of-degrees random variables.
Each of these random variables is the sum of many correlated indicator random variables.
To obtain tight bounds, we take advantage of the structure of the correlation by analyzing the cycle-path decomposition of the imposter matching relative to the true matching.
This decomposition is explained in Section~\ref{section:lifted}.
In Lemma~\ref{lemma:A-ub}, whose proof is in Appendix~\ref{app:gf}, we bound the generating function for the sum-of-degrees random variable using combinatorial arguments.
This allows us to find conditions under which the tails of their distributions behave like the tails of Poisson random variables.

Our converse proof is based on the concept of list estimation or list decoding.
It has two main components.
First, we find a relationship between the number of automorphisms of $\bfG_a \wedge_{\bfmu} \bfG_b$ and the list length for any list estimator that succeeds with high probability.
Second, we use the fact that sufficiently sparse \ErRe graphs have many isolated vertices to obtain a lower bound on the number of automorphisms of $\bfG_a \wedge_{\bfmu} \bfG_b$.

\section{Achievability}
\label{section:achievability}
\subsection{Weak $k$-core alignments}
\label{section:weak}
The property discussed at the start of Section~\ref{subsection:core} that leads to the uniqueness of the $k$-core has the following analogue for $k$-core alignments.
If $\delta(G_a \wedge_{\mu} G_b) \geq k$, $\delta(G_a \wedge_{\mu'} G_b) \geq k$, and $\mu \cup \mu'$ is a matching, then $\delta(G_a \wedge_{\mu \cup \mu'} G_b) \geq k$.
The graphs $\G{\mu}$ and $\G{\mu'}$ are induced subgraphs of $\G{(\mu \cup \mu')}$.
Thus each vertex in $\G{(\mu \cup \mu')}$ has a degree that is at least as large as its degree in the subgraph.

Because $\mu \cup \mu'$ is not guaranteed to be a matching, there may be more than one $k$-core alignment of $G_a$ and $G_b$.
For example, if both $G_a$ and $G_b$ are complete graphs with $n$ vertices, every bijection between $V(G_a)$ and $V(G_b)$ is an $n$-core alignment.

The $k$-core alignment estimator can make an error when some matching other than $\bfmu$ is a $k$-core alignment.
To analyze this event, we introduce weak $k$-core alignments.
\begin{definition}
  Let $\operatorname{deg}_{G} : V(G) \to \N$ be the degree function for the graph $G$.
  Let $\mu$ and $\mu^*$ be matchings of $V(G_a)$ and $V(G_b)$ and let
  \[
M(\mu,\mu^*;G_a,G_b) = \sum_{v \in \mu \setminus \mu^*} \operatorname{deg}_{G_a \wedge_{\mu} G_b}(v).
  \]
  A matching $\mu$ is a weak $k$-core alignment of $G_a$ and $G_b$ relative to $\mu^*$ if $M(\mu,\mu^*;G_a,G_b) \geq k |\mu \setminus \mu^*|$.

  Let $\bfM_{\mu,\mu^*} = M(\mu,\mu^*;\bfG_a,\bfG_b)$.
\end{definition}
Observe that if $\mu$ is a $k$-core alignment of $G_a$ and $G_b$, then it is also a weak $k$-core alignment of $G_a$ and $G_b$ relative to any matching $\mu^*$.
We have relaxed the property in two ways: first by only checking the vertices that are matched differently in $\mu$ than in $\mu^*$ and second by checking the average degree of these vertices rather than the minimum. 

All $\mu \subseteq \mu^*$ are trivially weak $k$-core alignments relative to $\mu^*$ (the sum is empty).
We will show that under certain conditions, every weak $k$-core alignment of $\bfG_a$ and $\bfG_b$ relative to $\bfmu$ is a subset of $\bfmu$.

\bcomment{
\begin{definition}
  Let $\mu \subseteq \mcU \times \mcU$ be a matching.
  Let $\id$ be the identity matching on $\mcU$.
  Fix $\mu^*$ and define the functions
\begin{align*}
  f(\mu) &= |\{(i,j) \in \mu : i = j\}| = |\mu \cap \id|\\
  g(\mu) &= |\{(i,j) \in \mu : i \neq j\}| = |\mu \setminus \id|\\
  h(\mu) &= |\{i : (i,j) \in \mu \setminus \id \} \cup \{j : (i,j) \in \mu \setminus \id \}|\\
  &= n - \max_{\mu' \supseteq \mu} f(\mu')
\end{align*}
\end{definition}

\begin{definition}
  A matching $\mu$ is $\mu^*$-maximal if it satisfies these equivalent conditions
  \begin{align*}
    \mu^* &\subseteq (\alpha(\mu) \times \mcU_b) \cup (\mcU_a \times \beta(\mu))\\
    \forall (i,j) \in \mu^* &: i \in \alpha(\mu) \wedge j \in \beta(\mu)
  \end{align*}
  i.e. if no pairs from $\mu^*$ can be added to $\mu$ without destroying the matching property.
\end{definition}
}
\begin{definition}
  A matching $\mu$ is $\mu^*$-maximal if no pairs from $\mu^*$ can be added to $\mu$ without destroying the matching property.
  More precisely, for all $(i,j) \in \mu^*$, either $i \in \alpha(\mu)$ or $j \in \beta(\mu)$.
  Let \TODO{decide on this notation}
  \[
    \mcM(\mu^*,d) = \{\mu : \text{$\mu$ is $\mu^*$ maximal and $|\mu \setminus \mu^*| = d$}\}.  
  \]
\end{definition}
This property allows us to show that a large number of matchings cannot be $k$-core alignments of $(\bfG_a,\bfG_b)$ because they are not $\bfmu$-maximal.
If it is possible to add any pairs from the true matching $\bfmu$ to $\mu$, then $\hat{\bfmu}_k \neq \mu$. 

\TODO{Explain outline}

\bcomment{
 and $|\mu^* \setminus \mu| = b$?
Of the $b$ pairs in $\mu^* \setminus \mu$, $b-a$ are only covered by a pair in $\mu$ in the same row, $b-a$ are only covered by a pair in the same column, and $2a-b$ are covered by both.
\begin{align*}
  \binom{n}{n-b,b-a,b-a,2a-b} a!
\end{align*}}

The following lemma allows us to bound the probability of error of the $k$-core alignment estimator.
\begin{lemma}
  \label{lemma:union-bound}
  Let $\mu^*$ be a bijection. Then
  \begin{equation*}
    \Pr\bigg[\bigvee_{\mu} (\delta(\bfG_a \wedge_{\mu} \bfG_b) \geq k) \wedge (\mu \not\subseteq \mu^*) \bigg] \leq \exp(n^2\xi) -1 \label{ub}
  \end{equation*}
  where the disjunction is over all $\mu \subseteq \mcU_a \times \mcU_b$ and 
  \begin{equation}
    \log \xi = \max_{d \geq 1}^n \max_{\mu \in \mcM(\bfmu,d)} \frac{1}{d} \log \Pr[\bfM_{\mu,\mu^*} \geq kd] \label{xi-def}
  \end{equation}
  % If $\Pr[M_{\mu,\mu^*} \geq k|\mu \setminus \mu^*|] \leq \xi^{|\mu \setminus \mu^*|}$ for all $\mu^*$-maximal matchings and $\xi \leq o(n^{-2})$, then $\Pr[\hat{\mu} \subseteq \mu^*] \geq 1 -o(1)$.
\end{lemma}
\begin{proof}
  If $|\mu \setminus \mu^*| = d$ and $\delta(\bfG_a \wedge_{\mu} \bfG_b) \geq k$, directly from the definition of $\bfM_{\mu,\mu^*}$ we have $\bfM_{\mu,\mu^*} \geq kd$.
  Any matching $\mu$ has a unique extension to a $\mu^*$-maximal matching that is produced by adding as many elements of $\mu^*$ as possible: $\mu' = \mu \cup (\mu^* \setminus (\alpha(\mu) \times \beta(\mu)))$.
  Then $\mu \setminus \mu^* = \mu' \setminus \mu^*$, $\mu' \in \mcM(\mu^*,d)$, and $\bfM_{\mu',\mu^*} \geq kd$.
  Thus the event in the statement of the lemma is equivalent to 
\begin{align*}
  \Pr\bigg[\bigvee_{d \geq 1}^n \bigvee_{\mu' \in \mcM(\mu^*,d)} \bfM_{\mu',\mu^*} \geq kd \bigg]
  &\leq \sum_{d \geq 1}^n \sum_{\mu' \in \mcM(\mu^*,d)} \Pr[\bfM_{\mu',\mu^*} \geq kd] \\
  &\leql{a} \sum_{d \geq 1}^{n} \frac{n^{2d}}{d!} \xi^d  \\
  &\leq \exp(n^2\xi) - 1
\end{align*}
where $(a)$ uses the bound $|\mcM(\mu^*,d)| \leq \frac{n^{2d}}{d!}$ and the definition of $\xi$ in \eqref{xi-def}.
Because $\mu^*$ is a bijection, each $\mu \in \mcM(\mu^*,d)$ is fully specified by $\mu \setminus \mu^*$.
There are $\binom{n}{d}$ choices of $\alpha(\mu \setminus \mu^*)$, $\binom{n}{d}$ choices of $\beta(\mu \setminus \mu^*)$, and $d!$ bijections between these sets, and $\binom{n}{d} \leq \frac{n^d}{d!}$.
\end{proof}

\subsection{Lifted matchings}
\label{section:lifted}
Given two matchings $\mu,\mu' \subseteq \mcU_a \times \mcU_b$, let $(\mu + \mu') : \mcU_a \times \mcU_b \to \N$ be the multisubset of $\mcU_a \times \mcU_b$ in which the multiplicity of each element is the sum of its multiplicity in $\mu$ and its multiplicity in $\mu'$.
In the bipartite multigraph $(\mcU_a,\mcU_b, \mu + \mu')$, all vertices have degree $0$, $1$, or $2$, so the multigraph is the union of paths and even-length cycles (including cycles of length two, which are pairs of parallel edges).
More precisely, the degree of $u_a \in \mcU_a$ is $\indset(u_a \in \alpha(\mu)) + \indset(u_a \in \alpha(\mu'))$ and the degree of $u_b \in \mcU_b$ is $\indset(u_b \in \beta(\mu)) + \indset(u_b \in \beta(\mu'))$.

A vertex pair $w \in \binom{\mu}{2}$ contains $0$, $1$ or $2$ elements from $\mu \setminus \mu^*$.
The others are from $\mu^* \cap \mu$.

%can be partitioned into three subsets: $\binom{\mu \setminus \mu^*}{2}$, $(\mu \setminus \mu^*) \times (\mu \cap \mu^*)$, and $\binom{\mu \cap \mu^*}{2}$.
This matters because
\begin{align}
  \bfM_{\mu,\mu^*}
  &= \sum_{v \in \mu \setminus \mu^*} \sum_{w \in \binom{\mu}{2}} \indset(v \in w)\cdot(\bfG_a \wedge_{\mu} \bfG_b)(w)\nonumber\\
  &= \sum_{w \in \binom{\mu}{2}} |w \cap (\mu \setminus \mu^*)| (\bfG_a \wedge_{\mu} \bfG_b)(w) \label{alt-M}
\end{align}

Because $\mu^*$ is a bijection and $\alpha(\mu^*) = \mcU_a$, $\alpha(\mu \cap \mu^*)$ and $\alpha(\mu^* \setminus \mu)$ partition the vertex set $\mcU_a$.
Similarly  $\beta(\mu \cap \mu^*)$ and $\beta(\mu^* \setminus \mu)$ partition the vertex set $\mcU_b$.
At the level of vertex pairs, we have partitions of $\binom{\mcU_a}{2}$ and $\binom{\mcU_b}{2}$ into three regions each:
\begin{align*}
  \mcV_{a,i} &= \left\{ w_a \in \binom{\mcU_a}{2} : |w_a \cap \alpha(\mu \cap \mu^*)| = 2-i \right\}\\
  \mcV_{b,i} &= \left\{ w_b \in \binom{\mcU_b}{2} : |w_b \cap \beta(\mu \cap \mu^*)| = 2-i \right\}.
\end{align*}
for $i \in \{0,1,2\}$.
The important fact for us is that both $\ell(\mu)$ and $\ell(\mu^*)$ match elements of $\mcV_{a,i}$ with elements of $\mcV_{b,i}$. 

The matchings $\ell(\mu)$ and $\ell(\mu^*)$ also decompose into a union of paths and cycles.
Because $\mu^*$ is a bijection, the length of each path is odd and the edges from $\mu$ are never the initial or final edges in a path.
Each of these paths and cycles stays within one of the three regions.

\begin{figure}
  \centering
  \begin{tikzpicture}[scale=0.8]
    \draw (5,11.5) node {$\mu^* = \{(0,5),(1,6),(2,7),(3,8),(4,9)\}$};
    \draw (5,10.5) node {$\mu = \{(0,6),(2,8),(3,7),(4,9)\}$};
    
%    \draw (8,0) node {$\mcU_a$};
    \coordinate[label=above:{$0$}] (v0) at (2,8.5);
    \coordinate[label=above:{$1$}] (v1) at (3,8.5);
    \coordinate[label=above:{$2$}] (v2) at (4,8.5);
    \coordinate[label=above:{$3$}] (v3) at (5,8.5);
    \coordinate[label=above:{$4$}] (v4) at (6,8.5);

%    \draw (8,8.5) node {$\mcU_b$};
    \coordinate[label=below:{$5$}] (v5) at (2,7);
    \coordinate[label=below:{$6$}] (v6) at (3,7);
    \coordinate[label=below:{$7$}] (v7) at (4,7);
    \coordinate[label=below:{$8$}] (v8) at (5,7);
    \coordinate[label=below:{$9$}] (v9) at (6,7);

    \draw (3.5,9.5) node {$\alpha(\mu^* \setminus \mu)$};
    \draw (  6,9.5) node {$\alpha(\mu^* \cap \mu)$};
    \draw (3.5,6) node {$\beta(\mu^* \setminus \mu)$};
    \draw (  6,6) node {$\beta(\mu^* \cap \mu)$};

    \draw (1.7,8.8) -- (1.7,9.1) -- (5.3,9.1) -- (5.3,8.8);
    \draw (1.7,6.7) -- (1.7,6.4) -- (5.3,6.4) -- (5.3,6.7);
    \draw (5.7,8.8) -- (5.7,9.1) -- (6.3,9.1) -- (6.3,8.8);
    \draw (5.7,6.7) -- (5.7,6.4) -- (6.3,6.4) -- (6.3,6.7);

%    \draw (11,3.5) node {$\binom{\mcU_a}{2}$};
    \coordinate[label=above:{$\{0,1\}$}] (v01) at ( 0,3.5);
    \coordinate[label=above:{$\{0,2\}$}] (v02) at ( 1,3.5);
    \coordinate[label=above:{$\{0,3\}$}] (v03) at ( 2,3.5);
    \coordinate[label=above:{$\{1,2\}$}] (v12) at ( 3,3.5);
    \coordinate[label=above:{$\{1,3\}$}] (v13) at ( 4,3.5);
    \coordinate[label=above:{$\{2,3\}$}] (v23) at ( 5,3.5);

    \coordinate[label=above:{$\{0,4\}$}] (v04) at ( 6.1,3.5);
    \coordinate[label=above:{$\{1,4\}$}] (v14) at ( 7.1,3.5);
    \coordinate[label=above:{$\{2,4\}$}] (v24) at ( 8.1,3.5);
    \coordinate[label=above:{$\{3,4\}$}] (v34) at ( 9.1,3.5);

%    \draw (11,2) node {$\binom{\mcU_b}{2}$};
    \coordinate[label=below:{$\{5,6\}$}] (v56) at ( 0,2);
    \coordinate[label=below:{$\{5,7\}$}] (v57) at ( 1,2);
    \coordinate[label=below:{$\{5,8\}$}] (v58) at ( 2,2);
    \coordinate[label=below:{$\{6,7\}$}] (v67) at ( 3,2);
    \coordinate[label=below:{$\{6,8\}$}] (v68) at ( 4,2);
    \coordinate[label=below:{$\{7,8\}$}] (v78) at ( 5,2);

    \coordinate[label=below:{$\{5,9\}$}] (v59) at ( 6.1,2);
    \coordinate[label=below:{$\{6,9\}$}] (v69) at ( 7.1,2);
    \coordinate[label=below:{$\{7,9\}$}] (v79) at ( 8.1,2);
    \coordinate[label=below:{$\{8,9\}$}] (v89) at ( 9.1,2);

    \draw (2.5,4.5) node {$\mcV_{a,2}$};
    \draw (2.5,1) node {$\mcV_{b,2}$};
    \draw (7.7,4.5) node {$\mcV_{a,1}$};
    \draw (7.7,1) node {$\mcV_{b,1}$};

    \draw (-0.5,3.9) -- (-0.5,4.2) -- (5.5,4.2) -- (5.5,3.9);
    \draw (5.6,3.9) -- (5.6,4.2)   -- (9.6,4.2) -- (9.6,3.9);
    \draw (-0.5,1.6) -- (-0.5,1.3) -- (5.5,1.3) -- (5.5,1.6);
    \draw (5.6,1.6) -- (5.6,1.3)   -- (9.6,1.3) -- (9.6,1.6);

    \draw[dotted] (v0) -- (v5);
    \draw[dotted] (v1) -- (v6);
    \draw[dotted] (v2) -- (v7);
    \draw[dotted] (v3) -- (v8);
    \draw[dotted] (v4) to[out=-105,in=105] (v9);

    \draw[dotted] (v01) -- (v56);
    \draw[dotted] (v02) -- (v57);
    \draw[dotted] (v03) -- (v58);
    \draw[dotted] (v12) -- (v67);
    \draw[dotted] (v13) -- (v68);
    \draw[dotted] (v23) to[out=-105,in=105] (v78);
    \draw[dotted] (v04) -- (v59);
    \draw[dotted] (v14) -- (v69);
    \draw[dotted] (v24) -- (v79);
    \draw[dotted] (v34) -- (v89);

    \draw (v0) -- (v6);
    \draw (v2) -- (v8);
    \draw (v3) -- (v7);
    \draw (v4) to[out=-75,in=75] (v9);

    \draw (v02) -- (v68);
    \draw (v03) -- (v67);
    \draw (v23) to[out=-75,in=75] (v78);
    \draw (v04) -- (v69);
    \draw (v24) -- (v89);
    \draw (v34) -- (v79);

  \end{tikzpicture}
  \caption{Illustration of the decomposition of $\mu + \mu^*$ and $\ell(\mu) + \ell(\mu^*)$ into cycles and paths. The matchings $\mu$ and $\ell(\mu)$ are drawn with solid lines and the bijections $\mu^*$ and $\ell(\mu^*)$ are draw with dotted lines. The sets $\mcV_{a,0}$ and $\mcV_{b,0}$ are empty. }
  \label{figure:cycle-path}
\end{figure}
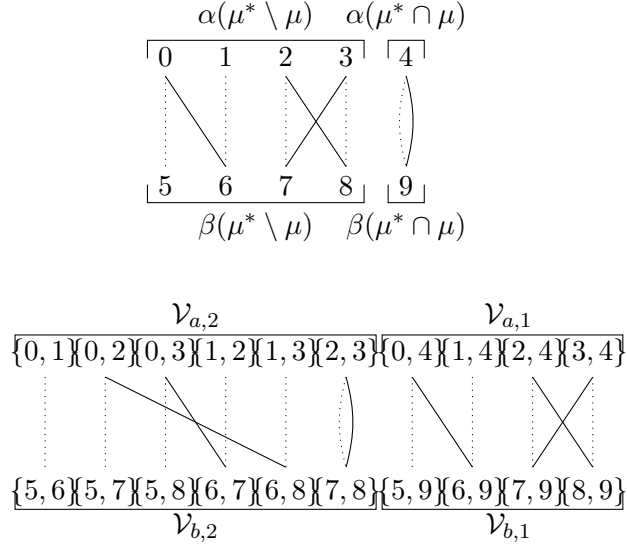

\begin{definition}
  For a matching $\mu$ and a bijection $\mu^*$, define the following.
  For $i \in \{0,1,2\}$, let $\nu_i = \ell(\mu) \cap (\mcV_{a,i} \times \mcV_{b,i})$ and $\nu_i^* = \ell(\mu^*) \cap (\mcV_{a,i} \times \mcV_{b,i})$.
  For $\ell \geq 1$, let $t_{i,\ell}^{\circ}$ be the number of cycles of length $2\ell$ and let $t_{i,\ell}$ be the number of paths of length $2\ell+1$ in the decomposition of $\nu_i + \nu_i^*$.  
\end{definition}
We have $\ell(\mu) = \nu_0 \cup \nu_1 \cup \nu_2$ and $\ell(\mu^*) = \nu_0^* \cup \nu_1^* \cup \nu_2^*$, so $t$ and $t^{\circ}$ capture the whole structure of $\ell(\mu) + \ell(\mu^*)$.

An example of this decomposition is illustrated in Figure~\ref{figure:cycle-path}.
The matchings $\ell(\mu)$ and $\ell(\mu^*)$ always have the same structure between $\mcV_{a,0}$ and $\mcV_{b,0}$.
Thus $t_{0,\ell}^{\circ} = 0$ for $\ell \geq 2$ and $t_{0,\ell} = 0$ for all $\ell$.
    The structure between $\mcV_{a,1}$ and $\mcV_{b,1}$ is $|\mu^*\cap \mu|$ copies of the structure between $\alpha(\mu^* \setminus \mu)$ and $\beta(\mu^* \setminus \mu)$ and consequently contains no cycles of length two, i.e. $t_{1,1}^{\circ} = 0$.
    Observe that between $\mcV_{a,2}$ and $\mcV_{b,2}$, a cycle of length two can only be produced by a cycle of length four in the region between $\alpha(\mu^* \setminus \mu)$ and $\beta(\mu^* \setminus \mu)$.
    Thus $t_{2,1}^{\circ} \leq |\mu \setminus \mu^*|/2$.

\subsection{Generating functions}

\begin{definition}
  Let $\gfM(z)$ be the generating function for the random variable $\bfM_{\mu,\mu^*}$:
  \[
    \gfM(z) = \E[z^{\bfM_{\mu,\mu^*}}].
  \]
\end{definition}

Let $p_{1*} = p_{10} + p_{11}$ and $p_{*1} = p_{01} + p_{11}$.
\begin{lemma}
  \label{lemma:A-ub}
  For a matching $\mu$ and a bijection $\mu^*$, if $\frac{p_{11}p_{00}}{p_{10}p_{01}} \geq 1$ then
  \begin{equation*}
    \log \gfM(z) \leq 
    t_{2,1}^{\circ}p_{11}(z^2-1)
    +\frac{\ttl}{4}(2\pmar(z^2 - 1) + p_{11}^2(z^2 - 1)^2)
  \end{equation*}
  where $\ttl = d(n-1) - 2t_{2,1}^{\circ}$ and $d = |\mu \setminus \mu^*|$.
\end{lemma}
The proof is given in Appendix~\ref{app:gf}.

\begin{lemma}
  \label{lemma:quadratic}
  Let $q_1 \geq 0$ and $q_2 \geq 0$. %$(q_0,q_1,q_2)$ be a probability distribution on $\{0,1,2\}$.
  Then
%Consider the g.f.
  \begin{align}
    \argmin_{z \geq 0} \exp(q_2(z^2-1) + q_1(z-1))z^{-\tau} &\leq \zeta^{\tau}\label{gf-ub}\\
    \zeta = \max \bigg(\sqrt{2} e \frac{q_1}{\tau}, 4e\pth{\frac{q_2}{\tau}}^{1/2}\bigg) .\nonumber
  \end{align}
\end{lemma}
%This follows from standard Chernoff bound methods.
The proof is given in Appendix~\ref{app:quadratic}.

\begin{lemma}
  \label{lemma:xi-ub}
  If $p$ satisfies conditions \eqref{density}, \eqref{sparsity}, and \eqref{correlation} and $k \geq \Omega(np_{11})$, then 
  \begin{equation}
    \max_{d \geq 1}^n \max_{\mu \in \mcM(\mu^*,d)} \frac{1}{d} \log \Pr[\bfM_{\mu,\mu^*} \geq kd] \leq - \omega(\log n)
  \end{equation}

\end{lemma}
The proof is given in Appendix~\ref{app:xi-ub}.
\begin{theorem}[\cite{luczak_size_1991} Theorem 2]
  \label{thm:k-core-size}
  Let $c = c(n) = (n-1)p(n)$.
  For every $\epsilon >0$ there is a constant $d$, such that for all $c(n) > d$ and $k(n) \leq c - c^{\frac{1}{2} + \epsilon}$, has the size of the $k$-core of a graph $G \sim G(n,p)$ is at least $n - n \exp(-c^{\epsilon})$ with probability $1-o(1)$.
\end{theorem}

\begin{proof}[Proof of Theorem~\ref{thm:main}]
  First we will show that $\bfG_a \wedge_{\bfmu} \bfG_b$ has a large $k$-core, so there is some $\hat{\mu} \subseteq \bfmu$ such that $\hat{\mu}$ is a $k$-core alignment and $|\hat{\mu}| \geq n(1-o(1))$.
  We have $(\bfG_a \wedge_{\bfmu} \bfG_b) \sim G(n,p_{11})$ and $n p_{11} \geq \omega(1)$.
  From the application of Theorem~\ref{thm:k-core-size} with $\epsilon = \frac{1}{4}$, for $k = n p_{11}(1-(np_{11})^{-\frac{1}{4}}) \geq n p_{11}(1 - o(1))$, $G_{\bfmu}$ has a $k$-core of size
  \[
    n(1-\exp(-(np_{11})^{\frac{1}{4}})) \geq n(1-e^{-\omega(1)}) \geq n(1-o(1))
  \]
  with probability $1-o(1)$.

  Now we will show that $\hat{\mu} \subseteq \mu^*$, i.e. every vertex pair in $\hat{\mu}$ is correct.
  From Lemma~\ref{lemma:union-bound}, the probability of error is at most $\exp(n^2\xi) - 1$ and from Lemma~\ref{lemma:xi-ub} we have $\xi \leq n^{-\omega(1)}$, so the probability that $\hat{\bfmu} \not \subseteq \bfmu$ is $o(1)$.
\end{proof}
\section{Converse}
\label{section:converse}
Recall from Section~\ref{section:estimation} that when we are trying to estimate a subset of $\bfmu$, the quality of a partial matching $\mu'$ depends on the list of bijections that extend it.

\begin{definition}
  Let $\bfX$ and $\bfY$ be random variables on $\mcX$ and $\mcY$ respectively.
  A list estimator for $\bfY$ given $\bfX$ is a function $S : \mcX \to 2^{\mcY}$.
  The estimator succeeds when $\bfY \in S(\bfX)$. 
\end{definition}

\begin{lemma}
  \label{lemma:recip}
  Let $\bfY$ be a random variable on a finite set $\mcY$ with distribution $P_Y$.
  Let $S \subseteq \mcY$ such that $|S| = \ell$.
  Then
  \begin{equation}
    \Pr[\bfY \in S] \leq \E\left[\min\pth{1,\frac{\ell}{|\{y' \in \mcY : P_Y(y') \geq P_Y(\bfY)\}|}}\right] \label{recip}
  \end{equation}
\end{lemma}
\begin{proof}
  Let $(p_0,p_1,\ldots)$ be the list of distinct probabilities that appear in $P_Y$, sorted from largest to smallest.
  Let $\mcY_i = \{y \in \mcY : P_Y(y) = p_i\}$.
  Let $S^*$ be a set of size $\ell$ that maximizes $\Pr[\bfY \in S]$.
  
  If $\sum_{i=0}^j |\mcY_i| \leq \ell$, then $\mcY_j \subseteq S^*$.
  If $\sum_{i=0}^j |\mcY_i| > \ell$, then $|\mcY_j \cap S^*| = \max(0,\ell - \sum_{i=0}^{j-1} |\mcY_i|)$.
  We have the inequality
  \begin{equation*}
     \pth{\ell - \sum_{i=0}^{j-1} |\mcY_i|}\pth{\sum_{i=0}^{j} |\mcY_i|}
    = \ell|\mcY_j| + \pth{\ell - \sum_{i=0}^{j} |\mcY_i|}\pth{\sum_{i=0}^{j-1} |\mcY_i|}
    \leq \ell|\mcY_j|,
  \end{equation*}
  so the fraction of $\mcY_j$ appearing in $S^*$ is
  \[
    \frac{|\mcY_j \cap S^*|}{|\mcY_j|} = \frac{\max(0,\ell - \sum_{i=0}^{j-1} |\mcY_i|)}{|\mcY_j|} \leq \frac{\ell}{\sum_{i=0}^{j} |\mcY_i|}.
  \]
  Observe that for $y \in \mcY_j$, $\{y' \in \mcY : P_Y(y') \geq P_Y(y)\} = \bigcup_{i=0}^j \mcY_i$.
  Thus for all $j$, the fraction of $\mcY_j$ appearing in $S^*$ is at most as large as the contribution of $\mcY_j$ to the right side of \eqref{recip}.
%  Let $p^* = \max_{y \in \mathcal{Y}} P_Y(y)$ and let $\mcY^* \subset \mcY$ contain the points with probability equal to $p^*$.
%  Then the contribution of $\mcY^*$ to the right side of \eqref{recip} is $|\mcY^*|p^*|\mcY^*|^{-1} = p^*$.  
\end{proof}

Let $\mu \subseteq \mcU_a \times \mcU_b$ be a matching and let $\pi \subseteq \mu \times \mu$ be a bijection.
Then we can extract another matching $\gamma(\pi) \subseteq \mcU_a \times \mcU_b$ as follows.
Observe that
\[
  \pi \subseteq (\mu \times \mu) \subseteq (\mcU_a \times \mcU_b) \times (\mcU_a \times \mcU_b)
\]
and define $\gamma(\pi) = \{(u_a,v_b) : ((u_a,u_b),(v_a,v_b)) \in \pi \}$.

Cullina and Kiyavash proved the following fact.
\begin{lemma}
  \label{lemma:intersection}
  Let $G_a$ and $G_b$ be graphs, let $\mu$ be a matching between their vertex sets, and let $\frac{p_{01}p_{10}}{p_{11}p_{00}} < 1$.
  For all $\pi \in Aut(G_a \wedge_{\mu} G_b)$,
  \begin{equation*}
    \Pr[\bfmu = \gamma(\pi) | (\bfG_a,\bfG_b) = (G_a,G_b)]
    \geq \Pr[\bfmu = \mu | (\bfG_a,\bfG_b) = (G_a,G_b)]
  \end{equation*}
\end{lemma}
\begin{proof}
  This follows immediately by combining Lemma II.2 and Lemma V.1 of \cite{cullina_exact_2017}.
\end{proof}

\begin{theorem}
  \label{thm:list-converse}
  Let $\bfS = S(\bfG_a,\bfG_b)$ be a list estimator for $\bfmu$.
  Then
  \[
    \Pr[\bfmu \in \bfS \wedge |\bfS| \leq \ell] \leq \E\left[\min\pth{1,\ell |\Aut(\bfG_a \wedge_{\bfmu} \bfG_b)|^{-1}} \right]
  \]  
\end{theorem}
\begin{proof}
  Suppose that $|S(G_a,G_b)| > \ell$ for some $(G_a,G_b)$.
  Then by instead using a shorter list in these cases, we can create some $S'$ such that $|S'(G_a,G_b)| \leq \ell$ for all $(G_a,G_b)$ and this can only increase $\Pr[\bfmu \in \bfS \wedge |\bfS| \leq \ell]$.
%  Thus we can assume $|S(G_a,G_b)| \leq \ell$ for all $(G_a,G_b)$ without loss of generality.
  
  Applying Lemmas~\ref{lemma:intersection} and \ref{lemma:recip} for all $(G_a,G_b)$ and then averaging over $(\bfG_a,\bfG_b)$ gives the claim.
\bcomment{
  \begin{align*}
    &\eq \Pr[\bfmu \in \bfS | (\bfG_a,\bfG_b) = (G_a,G_b)]\\
    &\leq \ell \max_{\mu} \Pr[\bfmu = \mu | (\bfG_a,\bfG_b) = (G_a,G_b)] \\
    &\leq \E \left[\left. \min\pth{1,\ell |\Aut(G_a \wedge_{\bfmu} G_b)|^{-1}} \right| (\bfG_a,\bfG_b) = (G_a,G_b)\right]
  \end{align*}}
\end{proof}

\begin{lemma}
  \label{lemma:isolated}
  If $\bfG \sim \ER(n,p)$ and $p \leq \mathcal{O}\pth{\frac{1}{n}}$, then $\bfG$ has $\Omega(n)$ isolated vertices with probability $1-o(1)$.
\end{lemma}
\begin{proof}
  The expected number of isolated vertices is $n(1-p)^{n-1} \geq n e^{\Omega(1)}$.
  By a standard use of the second moment method, with probability $1-o(1)$, the number of isolated vertices is within a factor of $1-o(1)$ of the mean.
\end{proof}
\begin{proof}[Proof of Theorem~\ref{thm:converse}]
  The estimated partial matching $\hat{\mu}$ can be interpreted as a list estimator.
  There are $(n-|\hat{\mu}|)!$ bijections $\mu$ that extend $\hat{\mu}$, i.e. $|\mu| = n$ and $\hat{\mu} \subseteq \mu$.
  Then  
  \begin{align*}
    \Pr[\bfmu \in \bfS \wedge |\bfS| \leq \epsilon(n)!]
    &\leql{a} \E \left[ \min\pth{1,\epsilon(n)! |\Aut(\bfG_a \wedge_{\bfmu} \bfG_b)|^{-1}} \right]\\
    &\leq \Pr \big[|\Aut(\bfG_a \wedge_{\bfmu} \bfG_b)| \leq \epsilon(n)!\big] + \frac{\epsilon(n)!}{|\Aut(\bfG_a \wedge_{\bfmu} \bfG_b)|}\\
    &\leql{b} o(1) + \frac{\epsilon(n)!}{\Omega(n)!}\\
    &\leq o(1).
  \end{align*}
where $(a)$ uses Theorem~\ref{thm:list-converse}, and $(b)$ uses the following argument.
  If a graph $G$ has $j$ isolated vertices, then $|\Aut(G)| \geq j!$.
  The true intersection graph $\bfG_a \wedge_{\bfmu} \bfG_b$ has distribution $\ER(n,p_{11})$, so from Lemma~\ref{lemma:isolated} and \eqref{converse-density}, $|\Aut(\bfG_a \wedge_{\bfmu} \bfG_b)| \geq (\Omega(n))!$ with probability $1-o(1)$. 
%  Step $(b)$ used the following inequality: for a random variable $X$ taking positive integer values and a threshold $x$, $\E [ \min\pth{1,\ell X^{-1}} ] \leq \Pr [ X \leq x] + \ell x^{-1}$.
\end{proof}
\appendix
\section{Proof of Lemma~\ref{lemma:A-ub}}
\label{app:gf}
\subsection{Generating function combinatorics}
\begin{definition}
  Define the following matrices indexed by $\{0,1\} \times \{0,1\}$:
  \[
    P = \twobytwo{p_{00}}{p_{01}}{p_{10}}{p_{11}} \quad
    Z = \twobytwo{1}{1}{1}{z}.
  \]
  For $\ell \geq 1$, define the generating functions
  \begin{align*}
    a_{\ell}(z) &= \onevec^T (PZ)^{\ell} P \onevec\\
    a_{\ell}^{\circ}(z) &= \tr ((PZ)^{\ell}).
  \end{align*}
\end{definition}

\begin{lemma}
  \label{lemma:cycle-path-decomp}
  \[
    \gfM(z) = \prod_{\ell \geq 1} a_{\ell}(z)^{t_{\ell,1}} a_{\ell}^{\circ}(z)^{t_{\ell,1}^{\circ}} a_{\ell}(z^2)^{t_{\ell,2}} a_{\ell}^{\circ}(z^2)^{t_{\ell,2}^{\circ}}
  \]
\end{lemma}
\begin{proof}
  From \eqref{alt-M}, we have
  \[
     \bfM_{\mu,\mu^*} = \sum_{w \in \nu_1} (\bfG_a \wedge_{\mu} \bfG_b)(w) + 2 \sum_{w \in \nu_2} (\bfG_a \wedge_{\mu} \bfG_b)(w).
   \]
   and these two terms are independent.
   From this we obtain
   \begin{multline*}
     \Pr[\bfG_a = G_a,\bfG_b=G_b] z^{M(\mu,\mu^*;G_a,G_b)} =\\
     \prod_{(w_a,w_b) \in \nu_1^*} p_{G_a(w_a),G_b(w_b)} \prod_{(w_a,w_b) \in \nu_2^*} p_{G_a(w_a),G_b(w_b)}
     \prod_{(w_a,w_b) \in \nu_1} z^{G_a(w_a) \wedge G_b(w_b)} \prod_{(w_a,w_b) \in \nu_2} z^{2(G_a(w_a) \wedge G_b(w_b))}
   \end{multline*}
   which is the contribution of a particular graph pair to the generating function for $\bfM_{\mu,\mu^*}$.
   Each $w_a$ or $w_b$ appears at most twice in this expression: once in a factor of $p_{G_a(w_a),G_b(w_b)}$ and up to once in a factor of $z^{(G_a(w_a) \wedge G_b(w_b))}$ or $z^{2(G_a(w_a) \wedge G_b(w_b))}$.
   Thus 
   \[
     \gfM(z) = \sum_{G_a,G_b} \Pr[\bfG_a = G_a,\bfG_b=G_b] z^{M(\mu,\mu^*;G_a,G_b)}
   \]
   factorizes based on the cycle and path decomposition of $\ell(\mu) + \ell(\mu^*)$.
   Because the value of $G_a(w_a)$ is used in at most places, the sum over $G_a(w_a) \in [2]$ can be interpreted as a matrix multiplication.
   The matrix that contributes the factor from $\ell(\mu)$ is $Z$ and the matrix that contributes the factor from $\ell(\mu^*)$ is $P$.  
%   Each cycle of length $\ell$ in $\mu_1$ contributes a factor $a_{\ell}^{\circ}(z)$: the sum of 
\end{proof}

\bcomment{
The random variables associated with both $a_1(z)^2$ and $a_2^{\circ}(z)$ have means of $2\pmar$, but the variance of the latter random variable is much larger.
On the other hand,
\[
  a_1(z)
  &= \onevec^TPZ_0P\onevec + \onevec^TPZ_1P\onevec\\
  &= 1 + \pmar(z-1)\\
\]
so the random variable associated with $a_1^{\circ}$ has a mean of $p_{11}$.
}
\begin{definition}
  For a sequence $f \in \{0,1\}^{\ell}$, let $k_1(f)$ be the number of ones in the sequence, $k_{11}(f)$ be the number of pairs of consecutive ones, and $k_{11}^{\circ}(f)$ be the number of pairs of consecutive ones counting the first and last position as consecutive.
  Define the generating functions
  \begin{align*}
    b_{\ell}(x,y) &= \sum_{f \in \{0,1\}^{\ell}} x^{k_1(f)} y^{k_{11}(f)} \\
    b_{\ell}^{\circ}(x,y) &= \sum_{f \in \{0,1\}^{\ell}} x^{k_1(f)} y^{k_{11}^{\circ}(f)}
  \end{align*}
\end{definition}

\begin{lemma}
  \label{lemma:a-to-b}
  For all $\ell \geq 1$,
  \begin{align*}
    a_{\ell}(z) &= b_{\ell}\pth{\pmar(z-1),\frac{p_{11}}{\pmar}}\\
    a_{\ell}^{\circ}(z) &= b_{\ell}^{\circ}\pth{\pmar(z-1),\frac{p_{11}}{\pmar}}
  \end{align*}
\end{lemma}
\begin{proof}
  Let $Z_0 = \twobytwo{1}{1}{1}{1}$ and $Z_1 = \twobytwo{0}{0}{0}{z-1}$. Then
  \begin{align*}
    a_{\ell}(z)
    &= \onevec^T (PZ)^{\ell} P \onevec\\
    &= \onevec^T (P(Z_0+Z_1))^{\ell} P \onevec\\
    &= \sum_{f \in \{0,1\}^{\ell}} \onevec^T \bigg(\prod_{i \in [\ell]} PZ_{f(i)}\bigg) P \onevec\\
    &\eql{a} \sum_{f \in \{0,1\}^{\ell}} (\pmar(z-1))^{k_1(f)} \pth{\frac{p_{11}}{\pmar}}^{k_{11}(f)}.
  \end{align*}
  Because $Z_0 = \onevec\onevec^T$, $\onevec^T \pth{\prod_{i \in [\ell]} PZ_{f(i)}} P \onevec$ is the product of terms of the form $\onevec^T(PZ_1)^jP\onevec = p_{01}(z-1)^jp_{11}^{j-1}p_{10}$.
  Each run of $j$ consecutive ones in $f$ contributes $j$ to $k_1(f)$ and $j-1$ to $k_{11}(f)$, which gives us $(a)$.

  The identity for $a_{\ell}^{\circ}$ follows analogously.
\end{proof}

\subsection{Inequalities}
\begin{lemma}
  \label{lemma:cycle-shorten}
  For all $\ell \geq 2$, $x \in \R$ and $y \geq 1$, $b_{\ell}^{\circ}(x,y) \leq b_2^{\circ}(x,y)^{\ell/2}$.
\end{lemma}
\begin{proof}
  Let $B = \twobytwo{1}{1}{x}{xy}$ and observe that
  \[
    b_{\ell}^{\circ}(x,y) = \sum_{f \in [2]^{\ell}} \prod_{i \in [\ell]} B_{i, (i+1) \bmod \ell} = \tr\pth{B^{\ell}}.
  \]
  Let $\lambda_0$ and $\lambda_1$ be the eigenvalues of $B$.
  When
  \begin{equation*}
    (\lambda_0- \lambda_1)^2
    = \tr(B)^2 - 4 \det(B)
%   = (1+xy)^2 + 4(x-xy)
    = y^2x^2 + (4-2y)x + 1
    \geq 0,
  \end{equation*}
  the eigenvalues are real.
  The discriminant of this quadratic is $(4-2y)^2 - 4y^2 = 16 - 16y$, which is negative when $y \geq 1$.
  Thus when $y \geq 1$, the eigenvalues are real for all $x$.
  We have the Jordan decomposition $B = C^{-1} \Lambda C$ where
  $\Lambda$ is upper triangular and has diagonal entries $\lambda_0$ and $\lambda_1$.
  Then 
  \[
    \tr(B^{\ell}) = \tr(\Lambda^{\ell}) = \lambda_0^{\ell} + \lambda_1^{\ell} \leq |\lambda_0|^{\ell} + |\lambda_1|^{\ell} \leq   (\lambda_0^2 + \lambda_1^2)^{\ell/2}
  \]
  from the standard inequality between $p$-norms.
%With \eqref{gf-iden}, this gives the claimed inequality.
\end{proof}

\begin{lemma}
  \label{lemma:path-to-cycle}
  For all $\ell \geq 1$, all $x \geq 0$, and all $y \geq 1$, $b_{\ell}(x,y)^2 \leq b_2^{\circ}(x,y)^{\ell}$.
\end{lemma}
\begin{proof}
  First, observe that for all $f \in \{0,1\}^{\ell}$, $k_{11}(f) \leq k_{11}^{\circ}(f)$.
  Thus for all $\ell \geq 2$, $j \in \N$, and $y \geq 1$, we have the stronger inequality $[x^k] b_{\ell}(x,y) \leq [x^j] b_{\ell}^{\circ}(x,y)$.
  Combining this with Lemma~\ref{lemma:cycle-shorten} gives the claim for $\ell \geq 2$.
  Finally, $b_1(x,y)^2 = (1+x)^2 \leq 1 + 2x + x^2y^2 = b_2^{\circ}(x,y)$.
\end{proof}

\begin{lemma}
  \label{lemma:a-ineqs}
  If $\frac{p_{11}p_{00}}{p_{10}p_{01}} \geq 1$, then for all $\ell \geq 1$, $a_{\ell}(z)^2 \leq a_2^{\circ}(z)^{\ell}$ and for all $\ell \geq 2$, $a_{\ell}^{\circ}(z)^2 \leq a_2^{\circ}(z)^{\ell}$.
\end{lemma}
\begin{proof}
  We have $y = \frac{p_{11}}{\pmar} \geq 1$ if and only if $\frac{p_{11}p_{00}}{p_{10}p_{01}} \geq 1$.
  Then the claim follows from Lemmas~\ref{lemma:a-to-b}, \ref{lemma:cycle-shorten}, and \ref{lemma:path-to-cycle}.
\end{proof}

\begin{proof}[Proof of Lemma~\ref{lemma:A-ub}]
Let $n' = |\mu|$.
We have
\begin{align}
   \gfM(z)%\nonumber\\
  &\eql{a} \prod_{\ell \geq 1} a_{\ell}(z)^{t_{1,\ell}} a_{\ell}^{\circ}(z)^{t_{1,\ell}^{\circ}} a_{\ell}(z^2)^{t_{2,\ell}} a_{\ell}^{\circ}(z^2)^{t_{2,\ell}^{\circ}}\nonumber\\
  &\leql{b} a_2^{\circ}(z)^{t_{1,1}/2} a_1^{\circ}(z)^{t^{\circ}_{1,1}} a_2^{\circ}(z^2)^{t_{2,1}/2} a_1^{\circ}(z^2)^{t^{\circ}_{2,1}} \nonumber\\
  &\eq \cdot \prod_{\ell \geq 2} (a_2^{\circ}(z)^{\ell/2})^{t_{1,\ell}+t_{1,\ell}^{\circ}} (a_2^{\circ}(z^2)^{\ell/2})^{t_{2,\ell}+t_{2,\ell}^{\circ}}\nonumber\\
%  &\eql{c} a_1^{\circ}(z^2)^{t_{2,1}^{\circ}} a_2^{\circ}(z)^{\ttl_1 /2} a_2^{\circ}(z^2)^{\ttl_2 /2} \label{one-two-cycles}
  &\eql{c} a_1^{\circ}(z^2)^{t_{2,1}^{\circ}} a_2^{\circ}(z)^{d(n'-d)/2} a_2^{\circ}(z^2)^{\pth{\binom{d}{2} - t_{2,1}^{\circ}} /2} \label{one-two-cycles}
%  &\leql{d} a_1^{\circ}(z)^{t_{1,1}^{\circ}} a_1^{\circ}(z^2)^{t_{2,1}^{\circ}} \exp( 2\pmar(z^2 - 1) + p_{11}^2(z^2 - 1)^2))^{\ttl/2}\\
\end{align}
where $(a)$ follows from Lemma~\ref{lemma:cycle-path-decomp}, $(b)$ follows from Lemma~\ref{lemma:a-ineqs}, and $(c)$ uses the facts $t_{1,1}^{\circ} = 0$, $\sum_{\ell \geq 1} (t_{1,\ell} + t_{1,\ell}^{\circ}) = d(n'-d)$, and $\sum_{\ell \geq 1} (t_{2,\ell} + t_{2,\ell}^{\circ}) = \binom{d}{2}$.

We can easily compute each of the factors in \eqref{one-two-cycles}:
\begin{align*}
  a_1^{\circ}(z)
  &= \tr(PZ_0) + \tr(PZ_1)\\
  &= 1 + p_{11}(z-1)\\
  a_2^{\circ}(z)
  &= \tr(PZ_0PZ_0) + 2\tr(PZ_1PZ_0) + \tr(PZ_1PZ_1)\\    
  &= 1 + 2\pmar(z-1) + p_{11}^2(z-1)^2.
\end{align*}
Now we will bound each factor:
\begin{align*}
  a_1^{\circ}(z^2)
%  &= 1 + p_{11}(z^2-1)\\
  &\leq \exp(p_{11}(z^2-1))\\
  a_2^{\circ}(z^2)
%  &= 1 + 2\pmar(z^2 - 1) + p_{11}^2(z^2 - 1)^2\\
  & \leq \exp(2\pmar(z^2 - 1) + p_{11}^2(z^2 - 1)^2)\\
  a_2^{\circ}(z)^2
%  &= (1 + 2\pmar(z - 1) + p_{11}^2(z - 1)^2)^2\\
  &\leq \exp( 4\pmar(z - 1) + 2 p_{11}^2(z - 1)^2)\\
  &=    \exp( 4(\pmar-p_{11}^2)(z - 1) + 2 p_{11}^2(z^2 - 1))\\
  &\leql{a} \exp( 2(\pmar-p_{11}^2)(z^2 - 1) + p_{11}^2(z^4 - 1))\\
  &= \exp( 2\pmar(z^2 - 1) + p_{11}^2(z^2 - 1)^2).
\end{align*}
where $(a)$ uses the inequality $x^2 - 1 = (x-1)^2 + 2(x-1) \geq 2(x-1)$.
Combining these with \eqref{one-two-cycles} and 
\[
  \frac{d(n'-d)}{4} + \frac{d(d-1) - 2 t_{2,1}^{\circ}}{4} = \frac{d(n'-1)-2 t_{2,1}^{\circ}}{4}
\]
gives the claimed bound.
\end{proof}
\section{Proof of Lemma~\ref{lemma:quadratic}}
\label{app:quadratic}
\begin{proof}
The optimal choice of $z$ satisfies
\begin{align}
  0 &= (2q_2z + q_1)\exp(q_2(z^2-1) + q_1(z-1))z^{-\tau}\nonumber\\
    &- \tau\exp(q_2(z^2-1) + q_1(z-1))z^{-\tau-1}\nonumber\\
  0 &= 2q_2z^2 + q_1z - \tau\label{quad}
\end{align}
The equation \eqref{quad} has one positive root and one negative root.
The positive root is
\[
  z^* = \frac{-q_1 + \sqrt{q_1^2 + 8\tau q_2}}{4q_2} = \frac{2\tau}{q_1 + \sqrt{q_1^2 + 8 \tau q_2}}.
\]
%The second expression avoid the cancellation between the terms in the numerator.
Because $q_1 \leq \sqrt{q_1^2 + 8 \tau q_2}$, we have the bounds
\begin{equation}
  \frac{\tau}{\sqrt{q_1^2 + 8\tau q_2}} \leq z^* \leq \frac{\tau}{q_1}. \label{z-bounds}
\end{equation}
Starting with one of the factors from the left side of \eqref{gf-ub}, we have
\begin{equation*}
  \exp\pth{q_2(z^2-1) + q_1(z - 1)}
  = \exp\pth{\frac{q_1}{2}z + \frac{\tau}{2} - q_2 - q_1}
  \leq e^{\tau - q_2 - q_1} \leq e^{\tau}
\end{equation*}
where we used \eqref{quad} to eliminate the $q_2z^2$ term, applied the upper bound from \eqref{z-bounds}, and used $q_1 \geq 0$ and $q_2 \geq 0$.
From the lower bound in \eqref{z-bounds},
\[
z^{-2} \leq \frac{q_1^2}{\tau^2} + \frac{8q_2}{\tau} \leq \max\pth{\frac{2q_1^2}{\tau^2},\frac{16q_2}{\tau}}
\]
so $\exp(q_2(z^2-1) + q_1(z-1))z^{-\tau} \leq \zeta^{\tau}$.
\end{proof}

\section{Proof of Lemma~\ref{lemma:xi-ub}}
\label{app:xi-ub}
\begin{proof}
  For $\mu \in \mcM(\mu^*,d)$ and $z > 0$,
  \begin{align*}
    \Pr[\bfM_{\mu,\mu^*} \geq kd]
    &\leq z^{-kd}\gfM(z)\\
    &\leq (z^2)^{-\tau}\exp(q_1(z^2 - 1) + q_2(z^4 - 1))
  \end{align*}
  where we have used Lemma~\ref{lemma:A-ub} and
  \begin{equation*}
    q_2 = \frac{\ttl}{4} p_{11}^2 \quad
    q_1 = t_{2,1}^{\circ}p_{11} + \frac{\ttl}{2}(\pmar - p_{11}^2) \quad
    \tau = \frac{dk}{2}.
  \end{equation*}
  Applying Lemma~\ref{lemma:quadratic} we obtain $\Pr[\bfM_{\mu,\mu^*} \geq kd] \leq \zeta^{\tau}$, where
  \[
    \zeta = \max \bigg(\sqrt{2} e \frac{q_1}{\tau}, 4e\pth{\frac{q_2}{\tau}}^{\frac{1}{2}}\bigg).
  \]
  Thus $\frac{1}{d} \log \Pr[\bfM_{\mu,\mu^*} \geq kd] \leq \frac{k}{2} \log \zeta$.
  We have
  \begin{align*}
    q_2 &\leq \frac{dn}{4} p_{11}^2\\
    q_1 &\leql{a} \frac{d}{2}p_{11} + \frac{dn}{2}(\pmar - p_{11}^2)\\
    \tau &\geq \Omega(dnp_{11})\\
    \frac{q_2}{\tau} &\leq \mathcal{O}(p_{11})\\
    \frac{q_1}{\tau} &\leq \mathcal{O}\pth{\frac{1}{n} + \frac{\pmar - p_{11}^2}{p_{11}}}
  \end{align*}
  where $(a)$ uses the fact that $t_{2,1}^{\circ}$ is equal to the number of cycles of length four in $\mu + \mu^*$, so it is at most $d/2$.
  From condition \eqref{correlation} and $p_{00} \geq \Omega(1)$ we have
  \[
    \frac{\pmar - p_{11}^2}{p_{11}} = \frac{p_{01}p_{10}}{p_{11}} + p_{01} + p_{10} \leq n^{-\Omega(1)}.
  \]
  To handle the case where $\zeta$ is equal to the first entry of the maximum, we have
  \begin{align*}
    \frac{k}{2} \log\pth{\frac{\tau}{\sqrt{2}eq_1}}
    &\geq \Omega(n p_{11}) \log (n^{\Omega(1)})\\
    &\geq \omega(\log n).
  \end{align*}
  In the second case, when $\omega\pth{\frac{1}{n}} \leq p_{11} \leq n^{-\Omega(1)}$, we have
  \begin{align*}
    \frac{k}{2} \log\pth{\pth{\frac{\tau}{16e^2q_2}}^{\frac{1}{2}}}
    &\geq \Omega(n p_{11}) \log \pth{\frac{1-o(1)}{8 e^2 p_{11}}}\\
    &\geq \omega(1) \log (n^{\Omega(1)})\\
    &\geq \omega(\log n)
  \end{align*}

  The function $f(x) = - x \log \pth{8e^2x}$ is increasing on the interval $0 \leq x \leq \frac{1}{8e^3}$, so condition \eqref{sparsity} can replace $p_{11} \leq n^{-\Omega(1)}$.
\end{proof}

\bibliography{IEEEabrv,deanon}
\bibliographystyle{IEEEtran}

\end{document}